\documentclass{elsart1p}

\usepackage{amsmath,amssymb}

\newenvironment{proof}{\noindent\textbf{Proof: }}{\mbox{}\hfill
    \makebox[0pt][l]{$\sqcup$}$\sqcap$\par\vspace{1ex} \vspace{9pt} }

\newtheorem{lemma}{Lemma}
\newtheorem{corollary}{Corollary}

\begin{document}

\begin{frontmatter}



\title{Finding a Feasible Flow in a Strongly Connected Network}

\author[princeton]{Bernhard Haeupler} and
\author[princeton,hp]{Robert E. Tarjan}
\address[princeton]{Department of Computer Science, Princeton University, 35 Olden Street, Princeton, NJ 08540-5233\\ \{haeupler,ret\}@cs.princeton.edu}
\address[hp]{HP Laboratories, Palo Alto CA}

\begin{abstract}
We consider the problem of finding a feasible single-commodity flow in a strongly connected network with fixed supplies and demands, provided that the sum of supplies equals the sum of demands and the minimum arc capacity is at least this sum. A fast algorithm for this problem improves the worst-case time bound of the Goldberg-Rao maximum flow method\cite{goldberg1998bfd} by a constant factor. Erlebach and Hagerup\cite{erlebach2002rft} gave an linear-time feasible flow algorithm. We give an arguably simpler one.
\end{abstract}


\end{frontmatter}

Let $G$ be a directed network with vertex set $V$, arc set $E$, non-negative arc capacities $c(v, w)$, and real-valued vertex \emph{imports} $b(v)$. A \emph{supply} is a positive import; a \emph{demand} is the absolute value of a negative import. A \emph{flow} on $G$ is a non-negative real-valued function on the arcs such that the flow on each arc does not exceed its capacity. Given a flow $f$, the \emph{balance} at a vertex $v$ is $e_f(v) = b(v) + \sum_{(u,v) \in E} f(u, v) - \sum_{(v, w) \in E} f(v, w)$. A flow is \emph{feasible} if all balances are zero. A necessary (but not sufficient) condition for the existence of a feasible flow is that the imports sum to zero, since for any flow the sum of balances equals the sum of imports. We assume henceforth that the sum of imports is indeed zero; that is, the sum of supplies equals the sum of demands. We further assume that the sum of supplies is one; if not, we can scale the imports and the arc capacities to make this true.\\

The problem of finding a feasible flow in a strongly connected network arises as a small part of the fast maximum flow method of Goldberg and Rao\cite{goldberg1998bfd}.  They gave the following linear-time algorithm to find a feasible flow if all arc capacities are at least two: \emph{``We choose an arbitrary vertex \ldots as \ldots root. Then we form an in-tree and an out-tree \ldots . We route all of the positive balances} [supplies] \emph{to the root using the in-tree, and we route the resulting flow excess from the root to the negative balances} [demands] \emph{using the out-tree.''} Goldberg and Rao also remarked that Knuth's\cite{knuth1974www} wheels-within-wheels decomposition of a strongly connected graph implies the existence of a feasible flow if all arc capacities are at least one, and that this flow can be computed in almost-linear time by building the wheels-within-wheels decomposition using a fast disjoint set data structure\cite{321884}. We note that this computation can be made linear-time by using depth-first search\cite{tarjan:146} and a linear-time disjoint set data structure\cite{gabow1983lta} to build the wheels-within-wheels decomposition.\\  

Being able to find a feasible flow for arc capacities at least one instead of two gives a constant-factor improvement in the worst-case time bound of the Goldberg-Rao maximum flow method, but as Goldberg and Rao observed, the wheels-within-wheels approach is more complicated than the simple expedient of routing flow through an in-tree and an out-tree. Erlebach and Hagerup\cite{erlebach2002rft} partially addressed this issue by giving a slightly complicated but linear-time algorithm for arc capacities at least one that does not use a disjoint set data structure but instead relies on properties of depth-first search exploited by Tarjan\cite{tarjan:146} in his strong components algorithm. We give a modification of the original linear-time Goldberg-Rao algorithm that works for arc capacities at least one. Our algorithm and its proof are simpler than those of Erlebach and Hagerup. To develop our algorithm, we first describe the Goldberg-Rao algorithm in more detail, then point out its capacity bottleneck and finally provide a way to overcome it. The Goldberg-Rao algorithm consists of the following three steps:\\

\noindent 1. Choose an arbitrary vertex $r$ as root. Find an in-tree $T$ rooted at $r$ containing all vertices with a supply and an out-tree $U$ rooted at $r$ containing all vertices with a demand. Initialize all arc flows to zero.\\

\noindent 2. Ignoring demands, move the supplies forward toward $r$ along the arcs of $T$.  Specifically, initialize the current supply $s(v)$ of each vertex $v$ in $T$ to be $\min\{0, b(v)\}$, and then process the nonroot vertices of $T$ in leaf-to-root order.  To process a vertex $v$ with parent $w$ in $T$, add $s(v)$ to $f(v, w)$ and to $s(w)$.\\

\noindent 3. Move the demands backward toward $r$ along the arcs of $U$. Specifically, initialize the current demand $d(v)$ of each vertex $v$ in $U$ to be $\min\{0, -b(v)\}$, and then process the nonroot vertices of $U$ in leaf-to-root order.  To process a vertex $w$ with parent $v$ in $U$, add $d(w)$ to $f(v, w)$ and to $d(v)$.\\

The flow increase on an arc is at most one (the sum of supplies) in Step 2 and at most one (the sum of demands) in Step 3, so the final flow satisfies the capacity constraints if every arc has capacity at least two. The only arcs requiring capacity more than one are those in both $T$ and $U$. If $(v, w)$ is such an arc, Step 3 increases the flow on $(v, w)$ by the final value of the current demand of $w$. We denote this value by $D(w)$; it is the sum of (original) demands of all descendants of $w$ in $U$. To prevent Step 3 from making $f(v,w)$ exceed $c(v, w)$, we restrict Step 2 to increase $f(v, w)$ by only up to $1 - D(v) \leq c(v, w) - D(w)$. This may cause Step 2 to leave residual supplies at vertices in $U$ other than $r$, but Step 3 (slightly modified) cancels these supplies, as we shall show. Here is the resulting version of the Goldberg-Rao algorithm, with the modifications to the original algorithm in {\bfseries boldface}:\\

\noindent 1$'$. Choose an arbitrary vertex $r$ as root. Find an in-tree $T$ rooted at $r$ containing all vertices with a supply and an out-tree $U$ rooted at $r$ containing all vertices with a demand. Initialize all arc flows to zero. {\bfseries For each nonroot vertex $\mathbf{v}$, compute $\mathbf{D(v)}$, the sum of demands of descendants of $\mathbf{v}$ in $\mathbf{U}$; if $\mathbf{v}$ is not in $\mathbf{U}$, $\mathbf{D(v) = 0}$.}\\

\noindent 2$'$. Ignoring demands, move the supplies forward toward $r$ along the arcs of $T$ {\bfseries as far as is safe}. Specifically, initialize the current demand $s(v)$ of each vertex $v$ in $T$ to be $\min\{0, b(v)\}$, and then process the nonroot vertices of $T$ in leaf-to-root order. To process a vertex $v$ with parent $w$ in $T$, {\bfseries let $\mathbf{x}$ = min\{$\mathbf{s(v), 1 - D(v)}$\}}; add $\mathbf{x}$ to $f(v, w)$ and to $s(w)${\bfseries, and subtract $\mathbf{x}$ from $\mathbf{s(v)}$}.\\

\noindent 3$'$. Move the demands backward toward $r$ along the arcs of $U$, {\bfseries canceling residual supplies}. Specifically, initialize the current {\bfseries net} demand $d(v)$ of each vertex $v$ to be $\min\{0, -b(v)\} \mathbf{- s(v)}$, and then process the nonroot vertices of $U$ in leaf-to-root order. To process a vertex $w$ with parent $v$ in $U$, add $d(w)$ to $f(v, w)$ and to $d(v)$.\\

\begin{lemma} Let $x$ be any vertex in $U$. Step 2$'$ maintains the invariant $S(x) = \sum_{v \in P} s(v) \leq D(x)$, where $P = \{v\ \vert\ v$ is a descendant of $x$ in $U$ already processed in Step 2$'$\}. 
\end{lemma}
\begin{proof}
Before any vertices are processed in Step 2$'$, $S(x) = 0 \leq D(x)$. Suppose the invariant holds just before Step 2$'$ processes $v$ with parent $w$ in $T$. Let unprimed and primed variables denote values just before and just after the processing of $v$.  If $S'(x) = S(x)$, then the invariant holds after the processing of $v$.  Otherwise, $s'(v) = S'(x) - S(x) > 0$. In this case $s(v) \leq 1 - S(x)$, which implies that $s'(v) \leq (1 - S(x)) - (1 - D(v)) = D(v) - S(x)$. Hence $S'(x) \leq D(v) \leq D(x)$; that is, the invariant holds after the processing of $v$. Thus Step 2$'$ maintains the invariant.
\end{proof}

\begin{corollary} When a vertex $w$ is processed in Step 3$'$, $d(w) \geq 0$.
\end{corollary}
\begin{proof}
Just after Step 2$'$, $S(w) \leq D(w)$ by Lemma 1. When $w$ is processed in Step 3$'$, $d(w) = D(w) - S(w)$.
\end{proof}

Each flow change in Step 2$'$ is an increase, and by Corollary 1 this is also true for each flow change in Step 3$'$. The increase of $f(v,w)$ during Step 2$'$ is at most $1 - D(v)$ and at most $D(w) \leq D(v)$ in Step 3$'$, summing to at most one. Thus the final $f$ is a flow if the arc capacities are at least one. Step 2$'$ leaves residual supplies only on vertices of $U$; after Step 3$'$, there is no residual supply or demand on any nonroot vertex. Thus the final flow is feasible. We conclude that the algorithm correctly finds a feasible flow on a strongly connected network with arc capacities of at least one. It is easy to implement the algorithm to run in linear time.




\end{document}